\documentclass[a4paper,reqno]{article}
\usepackage{amssymb}
\usepackage{latexsym}
\usepackage{amsmath}
\usepackage{euscript}
\usepackage{graphics,color}
\usepackage[all]{xy}
\usepackage[margin=3cm]{geometry}
\usepackage{MnSymbol}
\usepackage{authblk}
\usepackage{amsthm}
\usepackage{hyperref}

\newcommand{\be}{\begin{equation}}
\newcommand{\ee}{\end{equation}}
\newcommand{\ba}{\begin{eqnarray}}
\newcommand{\ea}{\end{eqnarray}}
\newcommand{\baa}{\begin{eqnarray*}}
\newcommand{\eaa}{\end{eqnarray*}}
\newcommand{\bb}{\begin{thebibliography}}
\newcommand{\eb}{\end{thebibliography}}
\newcommand{\ci}[1]{\cite{#1}}
\newcommand{\bi}[1]{\bibitem{#1}}
\newcommand{\lab}[1]{\label{#1}}
\newcommand{\re}[1]{(\ref{#1})}

\newcounter{my}

\newtheorem{theorem}{Theorem}[section]

\newtheorem{lemma}[theorem]{Lemma}

\theoremstyle{definition}

\newtheorem{remark}[theorem]{Remark}

\numberwithin{equation}{section}

\newcommand*{\email}[1]{\normalsize Email: natig@im.unam.mx }
\begin{document}
\title{The eigenvalues and eigenvectors of the  $5$D discrete Fourier transform 
number operator revisited}
\date{}

\author{Natig Atakishiyev}

\affil{Universidad Nacional Aut{\'o}noma de M{\'e}xico, \\
                Instituto de Matem{\'a}ticas,  Unidad Cuernavaca, \\
                Cuernavaca, 62210, Morelos, M{\'e}xico \\ \email{natig@im.unam.mx}}
\maketitle
\begin{abstract}
A systematic analytic approach to the evaluation of the eigenvalues and eigenvectors
of the $5$D discrete number operator $\mathcal{N}_5$ is formulated. This approach
is essentially based on the use of the symmetricity of $5$D discrete Fourier transform
operator $\Phi_5$ with respect to the discrete reflection operator $P_d$.
\end{abstract}

\section{Introduction}
Let me begin by recalling first that the eigenfunctions of the classical Fourier inegral
transform (FIT), associated with the eigenvalues ${\rm i}^n$, are explicitly given as 
\be
\psi_n(x): = H_n(x)\, {\exp(-x^2/2)},\quad n= 0,1,2,..., 
\label{psi}
\ee
where $ H_n(x) $ are Hermite polynomials. The functions $\psi_n(x)$ are usually
referred to as {\it Hermite functions} in the mathematical literature, whereas in
quantum mechanics they emerge as eigenfunctions of the Hamiltonian for the
linear harmonic oscillator, which is a self-adjoint differential operator of the 
second order (see, for example, \ci{LL}). It is well known that the  functions 
$\psi_n(x)$ are either symmetric or antisymmetric with respect to the reflection
operator $P$, defined on the full real line $x\in {\mathbb R}$ as $P\,x = -\, x$;
that is,
\be
P\,\psi_n(x) = \psi_n(-x) = (-1)^n\,\psi_n(x).
\label{Refpsi}
\ee
Recall also that the discrete (finite) Fourier transform (DFT) based on $N$ points 
is represented by an $N\times N$ unitary symmetric matrix $\Phi$  with entries 
\be
\Phi_{kl} = N^{-1/2} \, q^{kl}, \qquad k,l \in {\mathbb Z_N}:=\{0,1,2\dots, N-1\},
\label{DFT}
\ee
where $ q=\exp{\left(2 \pi {\rm i}\, / N\right)} $ is a primitive $N$-th root of unity
and $N$ is an arbitrary integer (see, for example, \cite{McCPar}-\cite{Ata}). 
The discrete analogue of the above mentioned reflection operator $P$, associated 
with the DFT operator (\ref{DFT}), is represented by the $N \times N$ matrix
\be
P_d:= C^{\intercal} J_N \equiv J_N\,C\,,
\label{Pd}
\ee
where $C$ is the {\it basic circulant permutation} matrix with entries $ C_{kl}
={\delta}_{k,l-1}$ and $J_N$ is the $N \times N$  {\lq}backward identity{\rq}
permutation matrix with ones on the secondary diagonal (see \cite{HJ}, pages 
26 and 28, respectively). Note that the matrix of the discrete reflection operator 
(\ref{Pd}) can be partitioned as
\be
P_d=\left[\begin{array}{cc}
              1                 &      0_{N-1}  \\
0^{\intercal}_{N-1} &     J_{N-1}   \\
\end{array}\right]
\label{partPd}\,,
\ee
where $0_m$ and $0^{\intercal}_{n}$ are $m$-row and $n$-column 
zero vectors, respectively.

It is readily verified that the DFT operator (\ref{DFT}) is $P_d$-symmetric,
that is, the commutator $[\Phi, P_d] := \Phi\,P_d - P_d\,\Phi=0$. Therefore,
similar to the continuous case (\ref{Refpsi}), the eigenvectors of the DFT
operator $\Phi$ should be either $P_d$-symmetric or $P_d$-antisymmetric.

The purpose of this work is to discuss some additional findings concerning 
symmetry properties of two finite-dimensional intertwining operators with 
the DFT matrix \re{DFT}. These operators are represented by matrices 
$A$ and $A^{\intercal}$ of the same size $N\times N$ such that the 
intertwining relations 
\be
A \,\Phi ={\rm i}\,\Phi A, \qquad A^{\intercal} \,\Phi 
= -\, {\rm i}\, \Phi \,A^{\intercal}, 
\lab{inter} 
\ee
are valid. The explicit form of the matrices $A$ and $A^{\intercal}$ is
\be
A= X + {\rm i} Y = X + D\,, \qquad A^{\intercal}= X - {\rm i} Y = X - D\,,
\label{intop}
\ee
where $X=diag\,(\mathsf{s}_0,\mathsf{s}_1,...,\mathsf{s}_{N-2},\mathsf{s}_{N-1}),$
\,$\mathsf{s}_n:=2\sin(2\pi n/N),\,n \in {\mathbb Z_N},$ and $Y= -\,\rm{ i}\,D=  
\rm{ i}\, (C^{\intercal} - C).$ The operators $X$ and $Y$ are Hermitian and play
the role of finite-dimensional analogs of the operators of the coordinate and
momentum in quantum mechanics, respectively. 

The intertwining operators $A$ and $A^{\intercal}$ have emerged in a paper \cite{MesNat} 
devoted to the problem of finding the eigenvectors of the DFT operator $\Phi$. They can be 
interpreted as discrete analogs of the quantum harmonic oscillator lowering and raising 
operators ${\bf a}= 2^{-1/2}\Big( x + \frac{d}{dx}\Big)$ and ${\bf {a}^{\dagger}}= 
2^{-1/2}\Big( x - \frac{d}{dx}\Big)$; their algebraic properties had been studied in
detail in \cite{AA2016} -\cite{AAZh}. In particular,  it was shown in \cite{AAZh} that
the operators $A$ and $A^{\intercal}$ form a cubic algebra $\mathcal{C}_q$ with 
$q$ a root of unity. This algebra is intimately related to the two other well-known 
realizations of the cubic algebra: the Askey-Wilson algebra \cite{Zhe_AW}--\cite{Tom2} 
and the Askey-Wilson-Heun algebra \cite{Bas_H}. Note also that from the intertwining 
relations (\ref{inter}) it follows at once that the operator $\mathcal{N}:=A^{\intercal} A$ 
commutes with the DFT operator $\Phi$, that is, $[\mathcal{N}\,,\Phi]=0.$ The discrete 
number operator $\mathcal{N}$ and the DFT operator $\Phi$ thus have the same 
eigenvectors and one can employ the former for finding an explicit form of the 
eigenvectors of the latter (see \cite{MesNat} for a more detailed discussion of this 
point). 

This idea that the discrete number operator $\mathcal{N}$ is the one that 
really governs the eigenvectors of the DFT operator $\Phi$, was first successfully
tested in \cite{AAMF} by considering the particular case of the $5$D DFT operator
$\Phi_5$. But the explicit form of the $4$ nonzero eigenvalues $\lambda_k, 1\leq k \leq 4$,
of the discrete number operator $\mathcal{N}_5$ have been found in  \cite{AAMF} by
using {\it Mathematica}. So it is the main goal of this work to formulate a systematic
analytic approach to the evaluation of the above-mentioned eigenvalues $\lambda_k$
without resorting to the help of any computer programs.

The lay out of the paper is as follows. In section 2 a detailed account is given on how 
one  can construct a $P_d$-symmetrized basis in the eigenspace ${\mathcal H}_5$ of
the discrete number operator $\mathcal{N}_5$, in terms of the eigenvectors of either 
the operator $X_5$, or the operator $Y_5$. In section 3 it is shown that the eigenspace 
${\mathcal H}_5$ with thus symmetrized basis splits into two $3$D and $2$D subspaces 
${\mathcal H}_3$ and ${\mathcal H}_2$; this remarkable fact is used then to find desired
explicit forms of the eigenvalues and eigenvectors of the discrete number operator 
$\mathcal{N}_5$. Finally, section 4 briefly outlines some further research directions 
of interest.

 \section{$5$D operators $X_5$ and $Y_5$ in the $P_d$-symmetrized basis}

This section begins by a quotation from \cite{AAZh}: {\it It is a remarkable 
fact that the operators $X$ and $Y$ are "classical" operators with nice spectral 
properties}. For the $5$D operator $X_5=diag\,(\mathsf{s}_0,\mathsf{s}_1,
\mathsf{s}_2,\mathsf{s}_3,\mathsf{s}_4)$, it is obvious because the spectrum 
of $X_5$ is 
\be
\lambda_n\,=\,\mathsf{s}_n={\rm i}(q^{-\,n} - q^{n}),\qquad n \in {\mathbb Z_5},
\lab{spX5}
\ee
where $q=\exp\,(2\pi {\rm i}/5)$ and we introduced for brevity $\mathsf{s}_n:=
2\sin\,(2\pi n/5)$. This indicates that the spectrum (\ref{spX5}) belongs to the 
class of the Askey--Wilson spectra of the type
\be
\lambda_n\,=\,C_1 q^{n} + C_2 q^{-\,n} + C_0\,.
\lab{AWsp}
\ee
The eigenvectors of the operator $X_5$ are represented by the Euclidean 
$5$-column orthonormal vectors $e_k$ with the components $(e_k)_l= 
\delta_{kl},\, k,l \in {\mathbb Z_5},$  that is,
\be
X_5\, e_k = \mathsf{s}_k\,e_k\,.
\lab{X5ev}
\ee
The spectrum of the matrix $Y_5$ belongs to the same Askey--Wilson family since the
operators $X_5$ and $Y_5$ are unitary equivalent, $Y_5=\Phi X_5 {\Phi}^{\dagger}$,
and hence isospectral \cite{AAZh}. Note that the spectrum of $X_5$ is simple, i.e., it
is nondegenerate. Also, from the unitary equivalence of the operators $X_5$ and $Y_5$
it follows that the eigenvectors of the latter operator are of the form
\be
Y_5\, {\epsilon}_k = \mathsf{s}_k\,{\epsilon}_k\,, \qquad{\epsilon}_n 
:= \Phi\, e_n = 5^{-1/2} \Big(1, q^n, q^{2n}, q^{3n}, q^{4n}\Big)^{\intercal}\,. 
\lab{Y5ev}
\ee
Let me draw attention now to the remarkable symmetry between the operators 
$X_5$ and $Y_5$: the operator $X_5$  is two-diagonal in the eigenbasis of the
operator $Y_5$, 
\be
X_5\, {\epsilon}_n = {\rm i}\,({\epsilon}_{n-1} - {\epsilon}_{n+1})\,, 
\lab{X5epsn}
\ee
whereas  the operator $Y_5$  is similarly two-diagonal in the eigenbasis of the
operator $X_5$,  
\be
Y_5\, e_n = {\rm i}\,(e_{n+1} - e_{n-1})\,.
\lab{Y5en}
\ee

\begin{remark} It may also be worth mentioning here that the $N$-column eigenvectors 
of the operator $Y$ for a general $N$,
\be
{\epsilon}_n = \Phi\, e_n = \sum_{k=0}^{N-1} {\Phi_{kn}\, e_k} 
= N^{-1/2} \Big(1, q^n, q^{2n},\dots, q^{(N-1)n}\Big)^{\intercal}\,, 
\lab{d_Phi_e}
\ee
form an orthonormal basis in the $N$-dimensional complex plane ${\mathbb C}^N$ and 
are frequently used therefore as building blocks of the discrete Fourier transform in 
applications (see, for example, p.130 in \cite{HarmAnal}, where the ${\epsilon}_n$
referred to as {\it discrete trigonometric functions}).
\end{remark} 
Since the operators $X$ and $Y$ generate a particular algebra, associated with the DFT 
operator $\Phi$ for arbitrary integer values of $N$,  one should use the eigenvectors of 
either the operator $X$, or the operator $Y$, as the most convenient basis for finding 
explicit forms of the eigenvectors of the operator $\Phi$. But we know that the eigenvectors 
of the operator $\Phi$ should be either $P_d$-symmetric, or $P_d$-antisymmetric, whereas 
the eigenvectors of both the operators $X$ and $Y$ do not reveal any symmetry property 
of this type. The point is that the reflection operator $P_d$ acts in the same way on both 
the eigenvectors $e_n$ and ${\epsilon}_m$, that is, 
\be
P_d\, e_n = e_{N-n},\quad e_N=e_0\,,\qquad\quad P_d \,{\epsilon}_n 
= {\epsilon}_{N-n},\quad {\epsilon}_N={\epsilon}_0\,.
\lab{Pdenepn}
\ee
Hence, the reflection operator $P_d$ does not transform the eigenvectors $e_0$ 
and ${\epsilon}_0$, and acts similarly by {\it cyclic permutation} on the other 
eigenvectors $e_m$ and ${\epsilon}_n$, with $1\leq m,n\leq N-1$. To overcome
this type of obstacle on the way of finding the eigenvectors of the operator $\Phi$, 
one thus needs to find first some $P_d$-symmetric bases, associated with both of 
the operators $X$ and $Y$. This can be achieved as follows.

Returning now to the case of the $5$D operators $X_5$ and $Y_5$, let us consider 
first unit column-vectors ${\widetilde e}_n,  n \in {\mathbb Z_5}$, defined in terms 
of the eigenvectors $e_n$ of the operator $X_5$ as
\be
{\widetilde e}_0=e_0, \qquad {\widetilde e}_k=\frac{1}{\sqrt 2}\,(e_k - e_{5-k}), 
\,\, k=1,2, \qquad {\widetilde e}_l=\frac{1}{\sqrt 2}\,(e_l + e_{5-l}),\,\, l=3,4.
\lab{symen}
\ee
The explicit componentwise forms of the thus $P_d$-symmetrized column-vectors 
${\widetilde e}_n$ are
\[
{\widetilde e}_0=(1,0,0,0,0)^{\intercal}, \quad {\widetilde e}_1=\frac{1}{\sqrt 2}\,
(0,1,0,0,-1)^{\intercal}, \quad {\widetilde e}_2=\frac{1}{\sqrt 2}\,(0,0,1,-1,0)^{\intercal}, 
\]
\be
{\widetilde e}_3=\frac{1}{\sqrt 2}\,(0,0,1,1,0)^{\intercal}, \quad
{\widetilde e}_4=\frac{1}{\sqrt 2}\,(0,1,0,0,1)^{\intercal}. 
\lab{cwsymen}
\ee
The interrelation (\ref{symen}) between the vectors ${\widetilde e}_n$ and the 
eigenvectors  $e_n$ of the operator $X_5$ can be written in the compact form as
${\widetilde e}_k= T e_k,\,  k \in {\mathbb Z_5}$, where the unitary matrix $T$ is
\be
T=\frac{1}{\sqrt 2}\left[\begin{array}{ccccc}
\sqrt 2 &  0  &  0  &  0 &  0 \\
    0      &  1  &   0  &  0 &  1 \\
    0      &  0  &   1  &  1 &  0 \\
    0      &  0  &  -1  &  1 &  0 \\
    0      & -1  &   0  &  0 &  1 
\end{array}\right]\,, \qquad T^{-1}= T^{\intercal}\,,\quad  T\, T^{-1} = T^{-1} T = I_5.
\lab{T}
\ee
Note that from the geometric point of view the matrix $T$  represents simply a product 
of two rotations by the same angle $\alpha=\pi/4$ in the $14$- and $23$-planes of the 
$5$D-space, that is, $T= R_{14}(\pi/4)R_{23}(\pi/4)$, where 
\be
R_{14}(\pi/4) =\left[\begin{array}{ccccc}
    1      &  0  &  0  &  0 &  0 \\
    0      &  \cos{\frac{\pi}{4}}   &   0  &  0 & \sin{\frac{\pi}{4}} \\
    0      &  0  &   1  &  0 &  0 \\
    0      &  0  &  0 &  1 &  0 \\
    0      & - \sin{\frac{\pi}{4}}  &   0  &  0 &\cos{\frac{\pi}{4}} 
\end{array}\right], \quad 
R_{23}(\pi/4)=\left[\begin{array}{ccccc}
    1 &  0  &  0  &  0 &  0 \\
    0      &  1  &   0  &  0 &  0\\
    0      &  0  &\cos{\frac{\pi}{4}}    &\sin{\frac{\pi}{4}} &  0 \\
    0      &  0  &  -\sin{\frac{\pi}{4}}  &\cos{\frac{\pi}{4}}  &  0 \\
    0      &  0  &   0  &  0 &  1 
\end{array}\right].
\lab{R14R23}
\ee

Similarly, let us introduce now $5$ orthonormal column-vectors ${\widetilde \epsilon}_n,  
n \in {\mathbb Z_5}$, defined in terms of the eigenvectors $\epsilon_n$ of the operator 
$Y_5$ as ${\widetilde \epsilon}_k= T {\epsilon}_k,\,  k \in {\mathbb Z_5}$, with the same 
matrix $T$ as in (\ref{T}). Then the explicit forms of these $P_d$-symmetrized $5$-vectors 
${\widetilde \epsilon}_n$ are
\[
{\widetilde \epsilon}_0= {\epsilon}_0= \frac{1}{\sqrt 5}(1,1,1,1,1)^{\intercal}, \qquad 
{\widetilde \epsilon}_1=\frac{1}{\sqrt 2}\,({\epsilon}_1 - {\epsilon}_4)=
\frac{\rm i}{\sqrt {10}}(0,\mathsf{s}_1,\mathsf{s}_2,-\mathsf{s}_2,
-\mathsf{s}_1)^{\intercal}, 
\]
\[
{\widetilde \epsilon}_2=\frac{1}{\sqrt 2}\,({\epsilon}_2 - {\epsilon}_3)=
\frac{\rm i}{\sqrt {10}}(0,\mathsf{s}_2,-\mathsf{s}_1,\mathsf{s}_1,
-\mathsf{s}_2)^{\intercal},
\]
\[
{\widetilde \epsilon}_3=\frac{1}{\sqrt 2}\,({\epsilon}_2 + {\epsilon}_3)=
\frac{\rm i}{\sqrt {10}}(\mathsf{c}_0,\mathsf{c}_2,\mathsf{c}_1,\mathsf{c}_1,
\mathsf{c}_2)^{\intercal},
\]
\be
{\widetilde \epsilon}_4=\frac{1}{\sqrt 2}\,({\epsilon}_1 + {\epsilon}_4)=
\frac{\rm i}{\sqrt {10}}(\mathsf{c}_0,\mathsf{c}_1,\mathsf{c}_2,\mathsf{c}_2,
\mathsf{c}_1)^{\intercal},
\lab{symeps}
\ee
where $\mathsf{c}_n:=2\cos\,(2\pi n/5)$. 

It remains only to recall that if $Z_{kl}:=(e_k, Z e_l)$ represent the matrix 
elements of the operator (matrix) $Z$ in the basis $e_n$, then the matrix
$\tilde Z:=T Z T^{-1}$  represents the matrix elements of the same 
operator $Z$ in the basis ${\widetilde e}_n = T e_n$. Hence the explicit
forms of the operators $X_5$ and $Y_5$ in the $P_d$-symmetrized basis
${\widetilde e}_n$ are 
\[
{\tilde X}_{5} = T X_5  T^{-1} = - \left[\begin{array}{ccccc}
    0      &           0          &            0       &        0           &         0            \\
    0      &           0          &            0       &        0           & \mathsf{s}_1 \\
    0      &           0          &            0       &\mathsf{s}_2&        0             \\
    0      &           0          &\mathsf{s}_2&       0            &        0             \\
    0      &\mathsf{s}_1  &           0        &       0            &        0
\end{array}\right]=
- \left[\begin{array}{cc}
               0_{33}         &       x_{32}   \\
    x^{\intercal}_{32} &       0_{22}  \\
 \end{array}\right], \qquad
 x_{32}:=\left[\begin{array}{ccc}
            0         &         0             \\
            0         & \mathsf{s}_1  \\
 \mathsf{s}_2 &         0            \\
 \end{array}\right], 
\]
\be
{\tilde D}_{5} = T D_5  T^{-1} =\left[\begin{array}{ccccc}
    0        &    0  &    0   &   0 &  - {\sqrt 2} \\
    0        &    0  &    0   &  -1 &         0        \\
    0        &    0  &    0   &   1 &         1        \\
    0        &    1  &  - 1  &    0 &         0        \\
{\sqrt 2}&    0  &  - 1  &    0 &         0
\end{array}\right]=
\left[\begin{array}{cc}
             0_{33}         & d_{32}   \\
-\, d^{\intercal}_{32} & 0_{22} \\
 \end{array}\right], \qquad
d_{32}:=\left[\begin{array}{ccc}
   0  & - \,{\sqrt 2}\\
- 1   &        0         \\
  1   &        1         \\
 \end{array}\right], 
\lab{tildeX5Y5}
\ee
where $0_{nn}$ is the $n\times n$ zero matrix. 

Finally, from (\ref{tildeX5Y5}) it follows that the intertwining operators
$A_5$  and $A^{\intercal}_5$ in the basis ${\widetilde e}_n $ can be
partitioned as 
\[
{\tilde A}_{5}={\widetilde X}_{5} + {\widetilde D}_{5}=
\left[\begin{array}{cc}
                 0_{33}         &-\,a_{32}(s)\\
a^{\intercal}_{32}(- s)&   0_{22}     \\
\end{array}\right], \qquad
a_{32}(s):=\left[\begin{array}{ccc}
             0            &   {\sqrt 2}      \\
             1            & \mathsf{s}_1 \\
\mathsf{s}_2 - 1 &         - 1         \\
\end{array}\right],            
\]
\be
{\tilde A}^{\intercal}_{5}={\widetilde X}_{5} - {\widetilde D}_{5}
= \left[\begin{array}{cc}
                 0_{33}         & a_{32}(- s)\\
- a^{\intercal}_{32}(s)&   0_{22}     \\
\end{array}\right], \qquad
a_{32}(- s):=\left[\begin{array}{ccc}
             0               &   {\sqrt 2}      \\
             1               & - \mathsf{s}_1 \\
- \,\mathsf{s}_2 - 1 &         - 1         \\
\end{array}\right].            
\lab{A5tildeen}
\ee

To close this section, it may be worth drawing attention now to the particular
manner in which the $P_d$-symmetrization modifies explicit forms of the DFT
eigenvectors $f_k,\,k \in {\mathbb Z_5}$  in the basis  ${\widetilde e}_n$. 
It is obvious that the partitioning of the intertwining operators ${\tilde A}_{5}$ 
and ${\tilde A}^{\intercal}_{5}$ of the form (\ref{A5tildeen}) leads to the need 
to appropriately split the DFT eigenvectors $f_k,\,k \in {\mathbb Z_5}$ in the 
basis ${\widetilde e}_n$ into two components, that is, to represent the vectors 
${\widetilde f}_k:= T f_k$ as 
\be
{\widetilde f}_k = (\eta_k, \xi_k)^{\intercal}, \qquad \eta_k:= (x_0,x_1,x_2),
\qquad \xi_k:= (x_3, x_4)\,.
\lab{etaxi}
\ee
Since every  symmetric $5$D DFT eigenvector $f^{(s)}$ in the basis $e_n$ is 
of the form $f^{(s)} = (a,b,c,c,b)^{\intercal}$, whereas every  antisymmetric 
$5$D DFT eigenvector $f^{(a)}$ in the same basis $e_n$ is of the form $f^{(a)} 
= (0,b,c,-\,c,-\,b)^{\intercal}$, it turns out that
\be
{\widetilde f}^{\,(s)}: = T f^{(s)} =  (a,{\sqrt 2}b,{\sqrt 2}c,0,0)^{\intercal},
\qquad {\widetilde f}^{\,(a)}: = T f^{(a)} = - {\sqrt 2}\,(0,0,0,c,b)^{\intercal}.
\lab{updown}
\ee
This means that all $5$D DFT eigenvectors ${\widetilde f}_k$  in the basis  
${\widetilde e}_n$ are either of the $\eta$-type (that is, with vanishing 
lower  component $\xi_k$), or of the $\xi$-type (with the upper component 
$\eta_k=0$).

\section{DFT number operator  ${ N}_5$  in the $P_d$-symmetrized basis}

Having defined explicitly the matrices $A_5$ and $A^{\intercal}_5$ in the 
$P_d$-symmetrized basis ${\widetilde e}_n$ in the previous section, it is 
not hard to evaluate that the  discrete number operator  ${\mathcal N}_5 = 
A^{\intercal}_5\,A_5$ in the same basis ${\widetilde e}_n$ is of the 
following form
\be
{\tilde {\mathcal N}}_5= {\tilde A}^{\intercal}_{5}\,{\tilde A}_{5}
= \left[\begin{array}{cc}
   {\mathcal N}_{3}  &         0_{32}        \\
 0^{\intercal}_{32} & {\mathcal N}_{2}\\
\end{array}\right],
\lab{N5til}
\ee
where $0_{32}$ is the $3\times2$ zero matrix and ${\mathcal N}_3$ and 
${\mathcal N}_2$ are $3\times 3$ and $2\times 2$ full Hermitian matrices,
\be 
{\mathcal N}_{3}:=\left[\begin{array}{ccc}
                   2                    &       - {\sqrt 2}\, \mathsf{s}_1   &              - {\sqrt 2}                           \\
 - {\sqrt 2}\, \mathsf{s}_1 &          3 - \mathsf{c}_2            & \mathsf{c}_1\mathsf{s}_2 - 1         \\
           - {\sqrt 2}             & \mathsf{c}_1\mathsf{s}_2 - 1& 2(\mathsf{s}_2 + 2) - \mathsf{c}_1  \\
\end{array}\right],  \qquad          
{\mathcal N}_{2}:=\left[\begin{array}{cc}
 2(2 - \mathsf{s}_2) - \mathsf{c}_1 & \mathsf{c}_1\mathsf{s}_2 + 1 \\
     \mathsf{c}_1\mathsf{s}_2 + 1   &          5 - \mathsf{c}_2               \\
\end{array}\right],
\lab{uv}       
\ee
respectively. Thus the Fock space ${\mathcal H}_5$ of all eigenvectors of the 
discrete number operator ${\mathcal N}_5$ in the $P_d$-symmetrized basis 
${\widetilde e}_n$ splits into two $3$D and $2$D subspaces ${\mathcal H}_3$
and ${\mathcal H}_2$; the operator ${\tilde {\mathcal N}}_5$ represents in the 
eigenspace ${\mathcal H}_5$ the direct sum of the operators ${\mathcal N}_3$ 
and ${\mathcal N}_2$, that is, ${\tilde {\mathcal N}}_5 \,=\,{\mathcal N}_3
\,{\oplus}\,{\mathcal N}_2$.

One clarifying remark must be made at this point in connection with (\ref{N5til}). 
The point is that this formula reveals that the discrete number operator 
${\tilde {\mathcal N}}_5$ in the $P_d$-symmetrized basis ${\widetilde e}_n$ 
has $12$ zero matrix elements, whereas its counterpart ${\mathcal N}_5$  in 
the basis of the eigenvectors $e_n$ is represented by a $5$D matrix with 
$25$ nonzero entries. Note that it was possible to formulate such a remarkable 
transformation of the full matrix ${\mathcal N}_5$ into the {\it sparse matrix}\, 
${\tilde {\mathcal N}}_5$ only because of the $P_d$-symmetricity of the DFT
operator $\Phi_5$. Recall then that the well-known Fast Fourier Transform 
algorithm of Cooley and Tukey is based essentially on a factorization of the 
Fourier matrix into a product of {\it sparse matrices} (see, for example, 
\cite{{HarmAnal},{FFT}}). Thus it becomes clear now that Cooley and 
Tukey had been able to construct so ingeniously their highly efficient 
implementation of the DFT  only because of the $P_d$-symmetricity of the 
Fourier matrix, although they had never employed explicitly this fundamental 
symmetry property of the Fourier matrix.

From (\ref{N5til}) it is evident that the eigenvectors and eigenvalues of 
the operator ${\tilde {\mathcal N}}_5$ may be now defined in terms of 
the eigenvectors and eigenvalues of the operators ${\mathcal N}_3$ 
and ${\mathcal N}_2$ from the separate subspaces ${\mathcal H}_3$ 
and ${\mathcal H}_2$, respectively. In order to proceed to this task
under consideration, let me start first with the operator ${\mathcal N}_2$.

\begin{lemma} An arbitrary $2\times 2$ Hermitian matrix of the form
$ M= \left[\begin{array}{cc}
  a  &   b  \\
  b  &   d  \\
\end{array}\right]
$ 
can be written as a linear combination of the $2$D identity matrix $I_2$
and the  $2\times 2$ traceless matrix $M^{\prime}$,
\be
M= u\, I_2 + {M^{\prime}} = 
u\,\left[\begin{array}{cc}
  1  &   0  \\
  0  &   1  \\
\end{array}\right] +
\left[\begin{array}{cc}
  v  &    b  \\
  b  &  - v  \\
\end{array}\right]\,,
\lab{M}
\ee
where $2u= a + d$ and $2v= a - d$. The eigenvalues of the matrix $M^{\prime}$
are equal to 
\be
{\lambda}_{1,2} = \pm\, \Big(- \det {M}^{\prime}\Big)^{1/2} 
= \pm\,(v^2 + b^2)^{1/2},
\lab{spMtil}
\ee
whereas the eigenvalues of the matrix $M$ are 
\be
\mu_{1,2} = {\lambda}_{1,2} + (a + d)/2 =  (a + d)/2 \,\pm\,(v^2 + b^2)^{1/2}.
\lab{spM}
\ee
\end{lemma}
\begin{proof} Since 
\be
\det\,({M^{\prime}} - \lambda\, I_2) = 
\det \left[\begin{array}{cc}
 v - \lambda&           b            \\
        b         &  - \,v - \lambda\\
\end{array}\right]= {\lambda}^2 - v^2 - b^2\,,
\lab{lmpr}
\ee
the eigenvalues ${\lambda}_{1,2}$  of the matrix $M^{\prime}$ are roots of 
the quadratic equation ${\lambda}^2 - v^2 - b^2=0$; hence ${\lambda}_{1,2}
=\pm\,(v^2 + b^2)^{1/2}$ and formula (\ref{spMtil}) is proved. Then from
(\ref{M}) it follows at once that the eigenvalues of the matix $M$ are equal
to  ${\mu}_{1,2}={\lambda}_{1,2} + (a + d)/2$ and formula (\ref{spM}) is 
proved as well. 
\end{proof}
Evidently, the matrix ${\mathcal N}_{2}$ is of the same type as the matrix
$M$ from the  lemma above, with $a= 2(2 - \mathsf{s}_2) - \mathsf{c}_1$,
$b=\mathsf{c}_1\mathsf{s}_2 + 1$ and $d= 5 - \mathsf{c}_2$. This means
that in this particular case $u=5-\mathsf{s}_2$,  $v=\mathsf{c}_2 - \mathsf{s}_2
= \mathsf{c}_2 \,b$ and $v^2 + b^2= 2{\sqrt 5}\,(\mathsf{s}_2 + 2\mathsf{c}_1)
= {(\mathsf{s}_1)^2} \, b^2$. Thus from (\ref{spM}) it follows that the eigenvalues
of the matrix ${\mathcal N}_{2}$ are 
\be
{\mu}_1= 5 - \mathsf{s}_2 + \mathsf{s}_1\,b= 
5 + \mathsf{s}_2\,(\mathsf{s}_2 + \mathsf{c}_1),\qquad
{\mu}_2= 5 - \mathsf{s}_2 - \mathsf{s}_1\,b= 
\mathsf{s}_1\,(\mathsf{s}_1 + \mathsf{c}_2).
\lab{spN2}
\ee

\begin{remark} The equation which is solved to find eigenvalues of $n\times n$
matrix $M$ is usually interpreted as the equation for finding roots of the {\it
characteristic polynomial} in $\lambda$ of degree $n$,
\be
p_n(\lambda):= \det\,(\lambda\, I - M) = \lambda^n + c_1 \lambda^{n - 1} + 
c_2 {\lambda}^{n - 2} + \cdots +  c_{n - 1} \lambda + c_n\,,
\lab{chapol}
\ee
where $I$ is the $n\times n$ identity matrix and the coefficient $c_k$ is $(-1)^k$
times the sum of the determinants of all of the principal $k\times k$ minors of $M$
(in particular, $c_1= - \,{\mathrm{trace}\,( M)}$ and  $c_n= (- 1)^n \det M$).
The lemma $3.1$ has been employed in order to reduce the {\it characteristic 
equation} $p_2(\lambda)= \lambda^2 + c_1 \lambda + c_2= 0$ for the matrix 
${\mathcal N}_{2}$ to the readily solvable equation for the matrix 
${\mathcal N}^{\prime}_2$, which is of the form $p_2(\lambda)= \lambda^2 + c_2= 0$.
\end{remark} 

Having defined the eigenvalues ${\mu}_1$ and ${\mu}_2$ of the matrix 
${\mathcal N}_{2}$, it is not hard to find eigenvectors of ${\mathcal N}_{2}$,
associated with those eigenvalues (\ref{spN2}). Indeed, note first that
\be
{\mathcal N}_{2}= (5 -  \mathsf{s}_2 )\, I_2 + {\mathcal N}^{\prime}_2,
\qquad  {\mathcal N}^{\prime}_2 =(1 + \mathsf{c}_1 \,\mathsf{s}_2 )
\left[\begin{array}{cc}
 \mathsf{c}_2 &               1         \\
          1           &- \,\mathsf{c}_2\\
\end{array}\right]\,,
\lab{N2pr}
\ee
where the eigenvalues of the matrix $\left[\begin{array}{cc}
 \mathsf{c}_2 &               1         \\
          1           &- \,\mathsf{c}_2\\
\end{array}\right]$ are $\pm\,\mathsf{s}_1$. Therefore to find the 
eigenvectors of the matrix ${\mathcal N}_{2}$, it is sufficient to 
determine the eigenvectors of the much simpler matrix 
$\left[\begin{array}{cc}
 \mathsf{c}_2 &               1         \\
          1           &- \,\mathsf{c}_2\\
\end{array}\right]$. So one readily derives  that
\[ 
\left[\begin{array}{cc}
\mathsf{c}_2 &               1        \\
          1          &-\,\mathsf{c}_2\\
\end{array}\right]\,
\left(\begin{array}{c}
     \mathsf{c}_1   \\
 1 + \mathsf{s}_2 \\
\end{array}\right) = \mathsf{s}_1
\left(\begin{array}{c}
     \mathsf{c}_1   \\
 1 + \mathsf{s}_2 \\
\end{array}\right), 
\]
\be
\left[\begin{array}{cc}
\mathsf{c}_2 &               1        \\
          1          &-\,\mathsf{c}_2\\
\end{array}\right]\,
\left(\begin{array}{c}
 1 + \mathsf{s}_2 \\
   - \mathsf{c}_1   \\
\end{array}\right) = - \,\mathsf{s}_1
\left(\begin{array}{c}
 1 + \mathsf{s}_2 \\
     - \mathsf{c}_1 \\
\end{array}\right). 
\lab{evN2I}
\ee
Thus explicit forms of the two linearly independent eigenvectors of 
the operator ${\mathcal N}_{2}$, associated with the eigenvalues 
(\ref{spN2}), are 
\be
\varphi_1:= ( \mathsf{c}_1, 1 +  \mathsf{s}_2)^{\intercal}\,,\qquad
\varphi_2:=(1 +  \mathsf{s}_2, - \,\mathsf{c}_1)^{\intercal}
\lab{xi1,2}, 
\ee
respectively. Note that the vectors $\varphi_1$ and $\varphi_2$ are 
essentially the same as the down-components of the antisymmetric 
eigenvectors ${\widetilde f}_1 = T f_1$ and ${\widetilde f}_3 = T f_3$ 
of the discrete number operator ${\tilde {\mathcal N}}_5$ in the 
$P_d$-symmetrized basis ${\widetilde e}_n$, where $f_1$ and $f_3$
have been already derived in \cite{AAMF} by employing {\it Mathematica};
that is 
\be
{\widetilde f}_1 = (0_3, \varphi_1)^{\intercal}, \qquad
{\widetilde f}_3 = (0_3, \varphi_1)^{\intercal}.
\lab{f1f3}
\ee

Turning now to the case of the matrix ${\mathcal N}_3$, one may likewise 
employ the polynomial $p_3 (\lambda)= \lambda^3 
+ c_1 \lambda^2 + c_2 \lambda + c_3$ in order to find first the eigenvalues 
of ${\mathcal N}_{3}$. It turns out that the determinant of the matrix 
${\mathcal N}_{3}$ is equal to zero,
\[ 
\det\,{\mathcal N}_{3}=\left\vert\begin{array}{ccc}
                     2                    &       - {\sqrt 2}\, \mathsf{s}_1  &              - {\sqrt 2}                            \\
 - {\sqrt 2}\, \mathsf{s}_1 &          3 - \mathsf{c}_2             & \mathsf{c}_1 \mathsf{s}_2 - 1         \\
           - {\sqrt 2}                & \mathsf{c}_1\mathsf{s}_2 - 1 & 2(\mathsf{s}_2 + 2) - \mathsf{c}_1 \\
\end{array}\right\vert  
 =\left\vert\begin{array}{ccc}
                     2                    &                                  0                                         &               0                      \\
 - {\sqrt 2}\, \mathsf{s}_1 &           \mathsf{s}_1 - 2\, \mathsf{c}_2              &    - \,\mathsf{s}_2 - 1   \\
           - {\sqrt 2}                & 3(\mathsf{c}_2 - \mathsf{s}_1) - \mathsf{c}_1 &  (\mathsf{s}_2 + 1)^2 \\
\end{array}\right\vert  
\]
\be        
 =2 \left\vert\begin{array}{cc}
         \mathsf{s}_1 - 2\, \mathsf{c}_2               &    - \,\mathsf{s}_2 - 1    \\
 3(\mathsf{c}_2 - \mathsf{s}_1) - \mathsf{c}_1&  (\mathsf{s}_2 + 1)^2 \\
\end{array}\right\vert  
 =2 (\mathsf{c}_2 - \mathsf{s}_1)\,\left\vert\begin{array}{cc}
           - 1               &           - 1               \\
 \mathsf{s}_2 + 1 & \mathsf{s}_2  + 1 \\
\end{array}\right\vert = 0\,.
\lab{detN3}       
\ee
Hence the characteristic equation for the matrix ${\mathcal N}_{3}$ reduces to
the form 
\be
\lambda (\lambda^2 + c_1 \lambda + c_2) = 0\,.
\lab{cheqN3}
\ee
Consequently, one of the eigenvalues of  the matrix ${\mathcal N}_{3}$ is 
$\lambda_0=0$, whereas the two remaining eigenvalues of  ${\mathcal N}_{3}$ 
are roots of the quadratic equation 
\be
\lambda^2 + c_1 \lambda + c_2 = 0\,,
\lab{redcheqN3}
\ee
where the coefficient $c_1= - \,{\mathrm{trace}\,( {\mathcal N}_{3})}= -2\,(5 + \mathsf{s}_2)$ 
and the coefficient $c_2$,  which represents the sum of the determinants of the three principal 
$2\times 2$ minors of ${\mathcal N}_{3}$,  is readily evaluated to be $c_2= 10 + (4 \,\mathsf{s}_2
+ 3) (\mathsf{s}_1)^2$. So one concludes that 
\be
\lambda_{1,2}= - \frac{\mathsf{c}_1 }{2} \pm \sqrt{\frac{(\mathsf{c}_1 )^2}{4} - \mathsf{c}_2}
\,=\,5 + \mathsf{s}_2 \,\pm ( \mathsf{c}_1\,\mathsf{s}_2 - 1)\, \mathsf{s}_1,
\lab{rts1,2}
\ee
upon taking into account the readily verified identity $2(2 - \mathsf{s}_1)= (\mathsf{s}_2 
+ \mathsf{c}_2)^2$.

Quite similar to the case of the matrix ${\mathcal N}_{2}$, the knowledge of the
explicit forms of the eigenvalues for the matrix ${\mathcal N}_{3}$ essentially
simplifies the task of defining the appropriate eigenvectors of ${\mathcal N}_{3}$
for each of those eigenvalues. Indeed, by looking for solutions of the equation
${\mathcal N}_{3}\,\phi_{\lambda} = \lambda\,\phi_{\lambda}$ in the form                 
$\phi_{\lambda} = (x_0, x_1, x_2)^{\intercal}$, one arrives simply at a system
of three homogeneous equations 
\be 
\begin{array}{cccccc}
(2 - \lambda)\,x_0 &  -    & {\sqrt 2}\, \mathsf{s}_1\,x_1& - &{\sqrt 2}\,x_2 & = 0, \\
- {\sqrt 2}\, \mathsf{s}_1\,x_0&  + & (3 - \mathsf{c}_2 - \lambda)\,x_1 & + & (\mathsf{c}_1 \mathsf{s}_2 - 1)\,x_2& = 0,\\
- {\sqrt 2}\, x_0 & + & (\mathsf{c}_1\mathsf{s}_2 - 1)\,x_1 & + & [2(\mathsf{s}_2 + 2) - \mathsf{c}_1 - \lambda]\,x_2& = 0,\\
\end{array} 
\lab{evN3}
\ee
for the components of the $3$D column-vector $\phi_{\lambda}$. 

1°. In the case of $\lambda_0=0$ the system (\ref{evN3}) reduces to
\be 
\begin{array}{cccccc}
       {\sqrt 2}\,x_0                   &  -  &     \mathsf{s}_1\,x_1     & -  &     x_2      & = 0, \\
- {\sqrt 2}\,\mathsf{s}_1\,x_0& + & (3 - \mathsf{c}_2)\,x_1 & + & (\mathsf{c}_1 \mathsf{s}_2 - 1)\,x_2& = 0,\\
       - {\sqrt 2}\, x_0                & + & (\mathsf{c}_1\mathsf{s}_2 - 1)\,x_1& + & [2(\mathsf{s}_2 + 2) - \mathsf{c}_1 ]\,x_2& = 0.\\
\end{array} 
\lab{ev0N3}
\ee
Eliminating the component $x_0$ by adding to the second equation in (\ref{ev0N3}) 
the first one, multiplied by $\mathsf{s}_1$, one arrives at the relation $x_1 = (1 + 
\mathsf{s}_2)\,x_2$. Substituting this relation back into the first equation enables one
to express the component $x_0$ via the component $x_2$ as
\[
 {\sqrt 2}\,x_0 = \mathsf{s}_1\,x_1  + x_2 = (\mathsf{s}_1 - 2\,\mathsf{c}_2)\,x_2.
\]
Taking into account that the system (\ref{evN3}) defines the eigenvector $\phi_0$
up to the multiplication by an arbitrary constant factor, one thus concludes that 
\be
\phi_0 = \Big(\mathsf{s}_1- 2\,\mathsf{c}_2, {\sqrt 2}\,(1 + \mathsf{s}_2), {\sqrt 2}\,\Big)^{\intercal}.
\lab{phi0}
\ee

2°. In the case of $\lambda_1=5 + \mathsf{s}_2 \,+ ( \mathsf{c}_1\,\mathsf{s}_2 - 1)
\,\mathsf{s}_1= 5 +  \mathsf{s}_2\,(\mathsf{s}_2 - \mathsf{c}_1)$, the system of
equations (\ref{evN3}) reduces to
\be 
\begin{array}{cccccc}
[\mathsf{c}_1(\mathsf{s}_2 + 1) - 5]\,x_0 &  -    & {\sqrt 2}\, \mathsf{s}_1\,x_1& - &{\sqrt 2}\,x_2 & = 0, \\
- {\sqrt 2}\, \mathsf{s}_1\,x_0&  + & [\mathsf{c}_1(\mathsf{s}_2 + 2) - 3]\,x_1 & + 
& (\mathsf{c}_1 \mathsf{s}_2 - 1)\,x_2& = 0,\\
- {\sqrt 2}\, x_0 & + & (\mathsf{c}_1\mathsf{s}_2 - 1)\,x_1 & - & (\mathsf{c}_2\,\mathsf{s}_1 + 3) \,x_2& = 0,\\
\end{array} 
\lab{ev1N3}
\ee
As in the previous case of $\lambda_0=0$, one eliminates the component $x_0$ by adding 
to the second equation in (\ref{ev1N3}) the third one, multiplied by $-\,\mathsf{s}_1$. This
leads to the relation 
\be
a\,x_1 + b\,x_2=0,\qquad a={\sqrt 5}\,\mathsf{s}_2 + 3\,\mathsf{c}_1 - 5, 
\qquad b = (4 - \mathsf{c}_1)\,\mathsf{s}_1  + 3\,{\mathsf{c}}_2 - 2,
\lab{ab}
\ee
interconnecting the components $x_1$ and $x_2$. It turns out that the coefficients
$a$ and $b$ in the relation (\ref{ab}) have a common factor, 
\be
a = \epsilon\, (2\,\mathsf{s}_2 - 2\mathsf{c}_1 + 3), \qquad
b = \epsilon\,(2\,\mathsf{s}_2 +1), \qquad \epsilon = -\,\mathsf{c}_2\,(\mathsf{c}_2\,\mathsf{s}_1 + 3).
\lab{eps}
\ee
Eliminating this common factor from the relation (\ref{ab}) reduces it to the
simpler form,
\be
(2\,\mathsf{s}_2 - 2\mathsf{c}_1 + 3)\,x_1 + (2\,\mathsf{s}_2 +1)\,x_2=0,
\lab{x1x2}
\ee
from which it follows at once that $x_1=-\,(2\,\mathsf{s}_2 +1) $ and 
$x_2=2\,(\mathsf{s}_2 - \mathsf{c}_1) + 3$. Substituting these values of 
$x_1$ and $x_2$ into the first equation in (\ref{ev1N3}), one finally finds 
that the component $x_0= {\sqrt 2}\,\mathsf{c}_1$. Thus the eigenvector
$\phi_1$, associated with the eigenvalue $\lambda_1$, has the form
\be
\phi_1 = \Big({\sqrt 2}\,\mathsf{c}_1,- \,2\,\mathsf{s}_2 - 1, \,2\,(\mathsf{s}_2 - \mathsf{c}_1) + 3\,\Big)^{\intercal}.
\lab{phi1}
\ee
3°. Finally, in the case of $\lambda_2=5 + \mathsf{s}_2 \,- ( \mathsf{c}_1\,\mathsf{s}_2 - 1)
\,\mathsf{s}_1= \mathsf{s}_1\,(\mathsf{s}_1 - \mathsf{c}_2)$, the system of equations 
(\ref{evN3}) reduces to
\be 
\begin{array}{cccccc}
\mathsf{c}_2\,(\mathsf{s}_1 + 1)\,x_0 & - & {\sqrt 2}\, \mathsf{s}_1\,x_1& - &{\sqrt 2}\,x_2 & = 0, \\
     - {\sqrt 2}\, \mathsf{s}_1\,x_0         &+ & (\mathsf{c}_2\,\mathsf{s}_1 + 1) \,x_1 & + & (\mathsf{c}_1 
\mathsf{s}_2 - 1)\,x_2& = 0,\\
- {\sqrt 2}\, x_0 & + & (\mathsf{c}_1\mathsf{s}_2 - 1)\,x_1 & - & {\mathsf{c}^2_1}\,(\mathsf{s}_1 + \mathsf{s}_1)\,x_2& = 0.\\
\end{array} 
\lab{ev2N3}
\ee
As in the previous case of $\lambda_1$, one eliminates the component $x_0$ by adding 
to the second equation in (\ref{ev2N3}) the third one, multiplied by $-\,\mathsf{s}_1$. 
This leads to the relation $x_1=x_2$, which then enables one to find from the first 
equation in (\ref{ev2N3}) that $x_0=- {\sqrt 2}\, \mathsf{c}_1\,x_1$. Thus the 
eigenvector $\phi_2$, associated with the eigenvalue $\lambda_2$, has the form
\be
\phi_2 = \Big(-\,{\sqrt 2}\,\mathsf{c}_1, \,1, \,1\Big)^{\intercal}.
\lab{phi2}
\ee
It remains only to add that the vectors $\phi_0$, $\phi_1$ and $\phi_2$,
associated with the eigenvalues $\lambda_0$, $\lambda_1$ and $\lambda_2$,
are essentially the same as the up-components of the three symmetric 
eigenvectors ${\widetilde f}_0 = T f_0$, ${\widetilde f}_4 = T f_4$ 
and ${\widetilde f}_2 = T f_2$  of the discrete number operator 
${\tilde {\mathcal N}}_5$ in the $P_d$-symmetrized basis ${\widetilde e}_n$, 
where $f_0$, $f_2$ and $f_4$ have been already derived in \cite{AAMF} by 
employing {\it Mathematica}; that is, 
\be
{\widetilde f}_0 = (\phi_0, 0_2)^{\intercal}, \qquad
{\widetilde f}_2 = (\phi_2, 0_2)^{\intercal}, \qquad
{\widetilde f}_4 = (\phi_1, 0_2)^{\intercal}.
\lab{f0f2f4}
\ee

\section{Concluding remarks}

To conclude this work, the following should be recalled first. Recently it has been proved 
that the {\lq}position{\rq} and {\lq}momentum{\rq} DFT operators $X$ and $Y$ form a 
special case of the Askey-Wilson algebra $AW(3)$ \cite{AAZh}. So it would be appropriate
to use the eigenvectors of either $X$, or $Y$, as a basis in the eigenspace of the discrete
number operator $\mathcal N$, that governs the eigenvectors of the DFT operator $\Phi$.
In this work it is shown that in the case of DFT this technique of employing the {\lq}position{\rq} 
$e_n$ and {\lq}momentum{\rq} $\epsilon_n$ eigenvectors for resolving an eigenvalue 
problem for the discrete number operator $\mathcal N$ is not applicable, unless those
eigenvectors are being symmetrized with respect to the discrete reflection operator $P_d$.
Therefore the $P_d$-symmetrization operator $T$ is found and a remarkable fact is 
established: it turns out that the matrix of the discrete number operator ${\mathcal N}_5$ 
in the $P_d$-symmetrized basis ${\widetilde e}_n= T e_n$ has only half of the number
of the nonzero entries of the same matrix in the initial basis $e_n$. This {\it sparsealization} 
of the discrete number operator ${\mathcal N}_5$ is shown to be essentially helpful for
finding explicit forms of the eigenvalues and eigenvectors of the operator ${\mathcal N}_5$.
Finally, I believe that just a bit more time is needed now to resolve an eigenvalue problem
for the DFT number operator ${\mathcal N}$ of a general dimension $N$.

\section{Acknowledgments}

I am  profoundly grateful to Alexei Zhedanov for the long illuminating discussions that
catalyzed the appearance of this work.  
\newpage

\bb{99}

\bi{LL} L.D.Landau, E.M.Lifshitz, Quantum Mechanics (Non-relativistic Theory), Pergamon
Press, Oxford, 1991.

\bi{McCPar} J.H.McClellan and T.W.Parks, {\it Eigenvalue and eigenvector decomposition 
of the discrete Fourier transform}, IEEE Trans. Audio Electroac., {\bf AU-20}, 66--74, 1972.

\bi{AusTol} L.Auslander and R.Tolimieri, {\it Is computing with the finite Fourier transform 
pure or applied mathematics?}  Bull. Amer. Math. Soc., {\bf 1}, 847--897, 1979.

\bi{DicSte} B.W.Dickinson  and K.Steiglitz, {\it Eigenvectors and functions of the discrete 
Fourier transform}  IEEE Trans. Acoust. Speech, {\bf 30}, 25--31, 1982.

\bi{Mehta} M.L.Mehta, {\it Eigenvalues and eigenvectors of the finite Fourier transform}, 
J. Math. Phys., {\bf 28}, 781--785, 1987.

\bi{Matv} V.B.Matveev, {\it Intertwining relations between the Fourier transform and
discrete  Fourier transform, the related functional identities and beyond}, Inverse Prob., 
{\bf 17}, 633--657, 2001.

\bi{Ata} N.M.Atakishiyev, {\it On $q$-extensions of Mehta's eigenvectors of the 
finite Fourier transform}, Int. J. Mod. Phys. A, {\bf 21}, 4993--5006, 2006.

\bi{HJ} R.A.Horn,  C.R.Johnson,  Matrix analysis, Cambridge University Press, 
Cambridge, 2009.

\bi{MesNat} M.K.Atakishiyeva and N.M.Atakishiyev, {\it On the raising and lowering difference 
operators for eigenvectors of the finite Fourier transform}, J. Phys: Conf. Ser., {\bf 597}, 
012012, 2015.

\bi{AA2016} M.K.Atakishiyeva and N.M.Atakishiyev, {\it On algebraic properties of the 
discrete raising and lowering operators, associated with the $N$-dimensional discrete 
Fourier transform}, Adv. Dyn. Syst. Appl., {\bf 11}, 81--92, 2016. 

\bi{4Open} M.K.Atakishiyeva, N.M.Atakishiyev and J.Loreto-Hern{\'a}ndez, {\it More 
on algebraic properties of the discrete Fourier transform raising and lowering operators}, 
4 Open, {\bf 2}, 1--11, 2019.

\bi{AAZh} M.K.Atakishiyeva, N.M.Atakishiyev and A.Zhedanov, {\it An algebraic interpretation 
of the intertwining operators associated with the discrete Fourier transform}, J. Math. Phys., 
{\bf 62}, 101704, 2021.

\bi{Zhe_AW} A.S.Zhedanov, {\it {\lq}{\lq}Hidden symmetry{\rq}{\rq} of Askey-Wilson 
polynomials}, Theoretical and Mathematical Physics {\bf 89}, 1146--1157, 1991.

\bi{Ter_AW} P.Terwilliger, {\it The Universal Askey-Wilson Algebra}, SIGMA {\bf 7}, 
069, 2011, arXiv:1104.2813.

\bi{Tom1} T.H.Koornwinder, {\it The relationship between Zhedanov's algebra AW(3) and
the double affine Hecke algebra in the rank one case}, SIGMA {\bf 3}, 063, 2007, 15 pp.; 
arXiv:math/0612730v4 [math.QA].

\bi{Tom2} T.H.Koornwinder, {\it Zhedanov's algebra AW(3) and the double affine Hecke 
algebra in the rank one case.II. The spherical subalgebra}, SIGMA {\bf 4}, 052, 2008, 17 pp.; 
arXiv:0711.2320v3 [math.QA].

\bi{Bas_H} P. Baseilhac, S. Tsujimoto, L. Vinet, and A. Zhedanov, {\it The Heun-Askey-Wilson 
Algebra and the Heun Operator of Askey-Wilson Type}, Annales Henri Poincar\'e {\bf 20} , 
3091--3112, 2019, arXiv: 1811.11407.

\bi{AAMF} M.K.Atakishiyeva, N.M.Atakishiyev and J.M{\'e}ndez Franco, {\it On a discrete 
number operator associated with the 5D  discrete Fourier transform}, Springer Proceedings 
in Mathematics\,\&\,Statistics, {\bf 164}, 273--292, 2016.

\bi{HarmAnal} M.C.Pereyra and L.A.Ward, Harmonic analysis: from Fourier to wavelets, 
AMS, Providence, Rhode Island, 2012.

\bi{FFT} K.R.Rao, D.N.Kim, J.J.Hwang, Fast Fourier Transform: Algorithms and Applications,
Springer, Dordrecht, Heidelberg, London, New York, 2010.

\end{thebibliography}

\end{document}